\documentclass[aps,prl,reprint,groupedaddress,longbibliography]{revtex4-2}
\pdfoutput=1
\usepackage[utf8]{inputenc} 
\usepackage[english]{babel} 
\usepackage[T1]{fontenc} 
\usepackage{amsmath} 
\usepackage{amssymb}
\usepackage{amsthm}
\usepackage{algorithm}
\usepackage{algpseudocode}
\usepackage{braket}
\usepackage{csquotes}
\usepackage{enumitem} 
\usepackage{eucal} 
\usepackage[scaled=.92]{helvet} 
\usepackage[scaled=0.95]{inconsolata}
\usepackage{microtype} 
\usepackage{thmtools}
\usepackage{thm-restate} 
\usepackage{verbatim}  
\usepackage{hyperref} 
\usepackage[usenames,dvipsnames,svgnames,table]{xcolor} 

\newcommand{\BB}[1]{\mathbb{#1}}

\newcommand{\C}[1]{\mathcal{#1}}

\DeclareMathOperator{\poly}{poly}
\DeclareMathOperator{\polylog}{polylog}

\DeclareMathOperator{\SQ}{SQ}
\DeclareMathOperator{\Q}{Q}

\newcommand{\eps}{\varepsilon}

\declaretheorem{theorem}
\declaretheorem[numbered=no,title=Theorem]{theorem*}

\theoremstyle{definition}

\declaretheorem[numbered=no,title=Proposition]{proposition*}

\declaretheorem[numbered=no,title=Corollary]{corollary*}
\declaretheorem[numbered=no]{definition}

\declaretheorem[numbered=no,title=Lemma]{lemma*}
\declaretheorem[sibling=theorem]{problem}
\declaretheorem[numbered=no,title=Problem]{problem*}

\begin{document}

\title{Quantum principal component analysis only achieves an exponential speedup \\because of its state preparation assumptions}


\author{Ewin Tang}
\email[]{ewint@cs.washington.edu}
\homepage[]{ewintang.com}
\affiliation{University of Washington}

\date{\today}

\begin{abstract}
A central roadblock to analyzing quantum algorithms on quantum states is the lack of a comparable input model for classical algorithms.
Inspired by recent work of the author [E.\ Tang, STOC'19], we introduce such a model, where we assume we can efficiently perform $\ell^2$-norm samples of input data, a natural analogue to quantum algorithms that assume efficient state preparation of classical data.
Though this model produces less practical algorithms than the (stronger) standard model of classical computation, it captures versions of many of the features and nuances of quantum linear algebra algorithms.
With this model, we describe classical analogues to Lloyd, Mohseni, and Rebentrost's quantum algorithms for principal component analysis [Nat.\ Phys.\ 10, 631 (2014)] and nearest-centroid clustering [arXiv:1307.0411].
Since they are only polynomially slower, these algorithms suggest that the exponential speedups of their quantum counterparts are simply an artifact of state preparation assumptions.
\end{abstract}



\maketitle

\section{Introduction}

Quantum machine learning (QML) has shown great promise toward yielding new exponential quantum speedups in machine learning, ever since the pioneering linear systems algorithm of Harrow, Hassidim, and Lloyd \cite{hhl09}.
Since machine learning (ML) routines often push real-world limits of computing power, an exponential improvement to algorithm speed would allow for ML systems with vastly greater capabilities.
While we have found many fast QML subroutines for ML problems since HHL \cite{rml14,wbl12,lgz16,zff15,bs17}, researchers have not been able to prove that these subroutines can be used to achieve an exponentially faster algorithm for a classical ML problem, even in the strongest input and output models \cite{childs09,Aar15}.
A recent work of the author \cite{tang18a} suggests a surprising reason why: even our best QML algorithms, with issues with input and output models resolved, fail to achieve exponential speedups.
This previous work constructs a classical algorithm matching, up to polynomial slowdown, a corresponding quantum algorithm for recommendation systems \cite{kp16}, which was previously believed to be one of the best candidates for an exponential speedup in machine learning \cite{Pre18}.
In light of this result, we need to question our intuitions and reconsider one of the guiding questions of the field: when is quantum linear algebra exponentially faster than classical linear algebra?

The main challenge in answering this question is not in finding fast classical algorithms, as one might expect.
Rather, it is that most QML algorithms are \emph{incomparable} to classical algorithms, since they take quantum states as input and output quantum states: we don't even know an analogous classical model of computation where we can search for similar classical algorithms~\cite{Aar15}.
The quantum recommendation system is unique in that it has a classical input, a data structure implementing QRAM, and classical output, a sample from a vector in the computational basis, allowing for rigorous comparisons with classical algorithms.

In our previous work we suggest an idea for developing classical analogues to QML algorithms beyond this exceptional case \cite{tang18a}:
\begin{quote}
    When QML algorithms are compared to classical ML algorithms in the context of finding speedups, any state preparation assumptions in the QML model should be matched with $\ell^2$-norm sampling assumptions in the classical ML model.
\end{quote}
In this work, we implement this idea by introducing a new input model, \emph{sample and query access} (SQ access), which is an $\ell^2$-norm sampling assumption.
We can get SQ access to data under typical state preparation assumptions, so fast classical algorithms in this model are strong barriers to their QML counterparts admitting exponential speedups.
To support our contention that the resulting model is the right notion to consider, we use it to dequantize two seminal and well-known QML algorithms, quantum principal component analysis \cite{lmr14} and quantum supervised clustering \cite{lmr13}.
That is, we give classical algorithms that, with classical SQ access assumptions replacing quantum state preparation assumptions, match the bounds and runtime of the corresponding quantum algorithms up to polynomial slowdown.
Surprisingly, we do so using only the classical toolkit originally applied to the recommendation systems problem, demonstrating the power of this model in analyzing QML algorithms.

From this work, we conclude that the exponential speedups of the quantum algorithms that we consider arise from strong input assumptions, rather than from the quantumness of the algorithms, since the speedups vanish when classical algorithms are given analogous assumptions.
In other words, in a wide swathe of settings, on \emph{classical data}, these algorithms do not give exponential speedups.
QML algorithms can still be useful for quantum data (say, states generated from a quantum system), though a priori it's not clear if they give a speedup in that case, since the analogous ``classical algorithm on quantum data'' isn't well-defined.

Our dequantized algorithms in the SQ access model provide the first formal evidence supporting the crucial concern about strong input and output assumptions in QML.
Based on these results, we recommend exercising care when analyzing quantum linear algebra algorithms, since some algorithms with poly-logarithmic runtimes only admit polynomial speedups.
BQP-complete QML problems, such as sparse matrix inversion \cite{hhl09} and quantum Boltzmann machine training \cite{kw17}, still cannot be dequantized in full unless BQP=BPP.
However, many QML problems that are not BQP-complete have strong input model assumptions (like QRAM) and low-rank-type assumptions (which makes sense for machine learning, where high-dimensional data often exhibits low-dimensional trends).
This regime is precisely when the classical approaches we outline here work, so such problems are highly susceptible to dequantization.
We believe continuing to explore the capabilities and limitations of this model is a fruitful direction for QML research.

\paragraph{Notation.}
$[n] := \{1,\ldots,n\}$.
Consider a vector $x \in \BB{C}^n$ and matrix $A \in \BB{C}^{m\times n}$.
$A_{i,*}$ and $A_{*,i}$ will refer to $A$'s $i$th row and column, respectively.
$\|x\|$, $\|A\|_F$, and $\|A\|$ will refer to $\ell^2$, Frobenius, and spectral norm, respectively.
$\ket{x} := \frac{1}{\|x\|}\sum_{i=1}^n x_i\ket{i}$ and $\ket{A} := \frac{1}{\|A\|_F}\sum_{i=1}^m \|A_{i,*}\|\ket{i}\ket{A_{i,*}}$ (where, by the previous definition, $\ket{A_{i,*}} = \frac{1}{\|A_{i,*}\|}\sum_{j=1}^n A_{i,j}\ket{j}$).
$A = \sum_{i=1}^{\min m,n} \sigma_iu_iv_i^\dagger$ is $A$'s singular value decomposition, where $u_i \in \BB{C}^m$, $v_i \in \BB{C}^n$, $\sigma_i \in \BB{R}$, $\{u_i\}$ and $\{v_i\}$ are sets of orthonormal vectors, and $\sigma_1 \geq \sigma_2 \geq \cdots \geq \sigma_{\min m,n} \geq 0$.
$A_{\sigma} := \sum_{\sigma_i \geq \sigma} \sigma_iu_iv_i^\dagger$ and $A_k := \sum_{i=1}^k \sigma_iu_iv_i^\dagger$ denote low-rank approximations to $A$.
We assume basic arithmetic operations take unit time, and $\tilde{O}(f) := O(f\log f)$.

\section{The dequantization model}
A typical QML algorithm works in the model where state preparation of input is efficient and a quantum state is output for measurement and post-processing.
(Here, we assume an ideal and fault-tolerant quantum computer.)
In particular, given a data point $x \in \BB{C}^n$ as input, we assume we can prepare copies of $\ket{x}$.
For $m$ input data points as a matrix $A \in \BB{C}^{m\times n}$, we additionally assume efficient preparation of $\ket{A}$, to preserve relative scale.
We wish to compare QML and classical ML on classical data, so state preparation usually requires access to this data and its normalization factors.
This informs the classical input model for our quantum-inspired algorithms, where we assume such access, and instead of preparing states, we can prepare measurements of these states.

\begin{definition}
We have $O(T)$-time \emph{sample and query access} to $x \in \BB{C}^n$ (notated $\SQ(x)$) if, in $O(T)$ time, we can query an index $i \in [n]$ for its entry $x_i$; produce an independent measurement of $\ket{x}$ in the computational basis; and query for $\|x\|$.
If we can only query for an estimate of the squared norm $\bar{x} \in (1\pm\nu)\|x\|^2$, then we denote this by $\SQ^\nu(x)$.
For $A \in \BB{C}^{m\times n}$, sample and query access to $A$ (notated $\SQ(A)$) is $\SQ(A_{1,*},\ldots, A_{n,*})$ along with $\SQ(\tilde{A})$ where $\tilde{A}$ is the vector of row norms, i.e.\ $\tilde{A}_i := \|A_{i,*}\|$.
\end{definition}

Sample and query (SQ) access will be our classical analogue to quantum state preparation.
As we noted previously \cite{tang18a}, we should be able to assume that classical analogues can efficiently measure input states: QML algorithms shouldn't rely on fast state preparation as the ``source'' of an exponential speedup.
The algorithm itself should create the speedup.

For typical instantiations of state preparation oracles on classical input, we can get efficient SQ access to input.
For example, given input in QRAM \cite{glm08}, a strong proposed generalization of classical RAM that supports state preparation, we can get log-dimension-time SQ access to input {\cite[Proposition~3.2]{tang18a}}.
Similarly, sparse and close-to-uniform vectors can be prepared efficiently, and correspondingly admit efficient SQ access~\footnote{See supplemental material for details on when sample and query access is possible; discussion on the relation of this work to the tensor networks literature; and full proofs for the results stated here. It includes Refs.~\cite{Pra14,grover02,vc06,White92,Vidal03,Vidal04,clopzm16,bc17,Eisert06,lvv15,Hoeffding94}}.
So, in usual QML settings, SQ assumptions are easier to satisfy than state preparation assumptions.

This leads to a model based on SQ access that we codify with the informal definition of ``dequantization''.
We say we \emph{dequantize} a quantum protocol
$\C{S}: O(T)\text{-time state preparation of }\ket{\phi_1},\ldots,\ket{\phi_c}\to \ket{\psi}$
if we describe a classical algorithm of the form
    $ \C{C}_\C{S}: O(T)\text{-time } \SQ(\phi_1,\ldots,\phi_c) \to \SQ^\nu(\psi) $
with similar guarantees to $\C{S}$ up to polynomial slowdown.
This is the sense in which we dequantized the quantum recommendation system in prior work \cite{kp16}.
In the rest of this article, we will dequantize two quantum algorithms, giving detailed sketches of the classical algorithms and leaving proofs of correctness to the supplemental material~\cite{Note1}.
These algorithms are applications of three protocols from our previous work \cite{tang18a} rephrased in our access model.

\section{Nearest-centroid classification}

Lloyd, Mohseni, and Rebentrost's quantum algorithm for clustering estimates the distance of a data point to the centroid of a cluster of points \cite{lmr13}.
The paper claims~\footnote{We were not able to verify the quantum algorithm (namely, the Hamiltonian simulation for preparing $\ket{\phi}$) as stated. For our purposes, we can make the minor additional assumption of efficient state preparation access to $\ket{\phi}$, which makes correctness obvious. When we refer to the quantum algorithm in this letter, we mean this version of it.} that this quantum algorithm gives an exponential speedup over classical algorithms.
We dequantize Lloyd et al's quantum supervised clustering algorithm \cite{lmr13} with only quadratic slowdown.
Though classical algorithms by Aaronson \cite{Aar15} and Wiebe et al.\ \cite[Section~7]{wks15} dequantize this algorithm for close-to-uniform and sparse input data, respectively, we are the first to give a general classical algorithm for this problem.

\begin{problem}[Centroid distance] \label{prob:sc}
Suppose we are given access to $V \in \BB{C}^{n\times d}$ and $u \in \BB{C}^d$.
Estimate $\|u - \frac{1}{n}\vec{1} V\|^2$ to $\eps$ additive error with probability $\geq 1-\delta$.
\end{problem}

Note that we are treating vectors as rows, with $\vec{1}$ the vector of ones.
Let $\bar{u} := \frac{u}{\|u\|}$ and let $\bar{V}$ be $V$, normalized so all rows have unit norm.
Both classical and quantum algorithms argue about $M \in \BB{R}^{(n+1)\times d}$ and $w \in \BB{R}^{n+1}$ instead of $u$ and $V$, where
$$
    M := \begin{bmatrix}\bar{u} \\ \frac{1}{\sqrt{n}}\bar{V}\end{bmatrix} \text{ and }
    w := \begin{bmatrix}\|u\| & -\frac{1}{\sqrt{n}}\tilde{V}\end{bmatrix}.
$$
Because $wM = u-\frac{1}{n}\vec{1}V$, we wish to estimate $\|wM\|^2 = wMM^\dagger w^\dagger$.
Let $Z := \|w\|^2 = \|u\|^2 + \frac{1}{n}\|V\|_F^2$ be an ``average norm'' parameter appearing in our algorithms.

\begin{theorem}[Quantum Nearest-Centroid \cite{lmr13}]
    Suppose that, in $O(T)$ time, we can (1) determine $\|u\|$ and $\|V\|_F$; or (2) prepare a state $\ket{u},\ket{V_1},\ldots,\ket{V_n}$, or $\ket{\tilde{V}}$.
    Then we can solve Problem~\ref{prob:sc} in $O(T\frac{Z}{\eps}\log\frac{1}{\delta})$ time.
\end{theorem}

The quantum algorithm proceeds by constructing the states $\ket{M}$ and $\ket{w}$, then performing a swap test to get $\ket{wM}$.
The swap test succeeds with probability $\frac{1}{Z}wMM^\dagger w^\dagger$, so we can run amplitude amplification to get an estimate up to $\eps$ error with $O(\frac{1}{\eps}\log\frac{1}{\delta})$ overhead.

Dequantizing this algorithm is simply a matter of dequantizing the swap test, which is done in Algorithm~\ref{alg:inner}.
Here, $\Q(y)$ is \emph{query access} to $y$, which supports querying $y$'s entries in $O(1)$ time, but not querying samples or norms.

\begin{algorithm}[H]
\caption{Inner product estimation} \label{alg:inner}
\begin{algorithmic}[1]
    \Statex {\bf Input:} $O(T)$-time $\SQ^\nu(x) \in \BB{C}^{n}$, $\Q(y) \in \BB{C}^n$
    \Statex {\bf Output:} an estimate of $\braket{x|y}$
    \State Let $s = 54\frac{1}{\eps^2}\log\frac{2}{\delta}$
    \State Collect measurements $i_1,\ldots,i_s$ from $\ket{x}$
    \State Let $z_j = x_{i_j}^\dagger y_{i_j}\frac{\|x\|^2}{|x_{i_j}|^2}$ for all $j \in [s]$ \Comment{$\BB{E}[z_j] = \braket{x|y}$}
    \State Separate the $z_j$'s into $6\log\frac{2}{\delta}$ buckets of size $\frac{9}{\eps^2}$, and take the mean of each bucket
    \State Output the (component-wise) median of the means
\end{algorithmic}
\end{algorithm}

From a simple analysis of the random variable $z_i$'s, we get the following result.

\begin{restatable}[{\cite[Proposition~4.2]{tang18a}}]{proposition}{inner}\label{prop:inner}
For $x,y \in \BB{C}^n$, given $\SQ^\nu(x)$ and $\Q(y)$, Algorithm~\ref{alg:inner} outputs an estimate of $\braket{x|y}$ to $(\eps+\nu+\eps\nu)\|x\|\|y\|$ error with probability $\geq 1-\delta$ in time $O(\frac{T}{\eps^2}\log\frac{1}{\delta})$.
\end{restatable}

For this protocol, quantum algorithms can achieve a quadratic speedup via amplitude estimation (but no more, by unstructured search lower bounds~\cite{bbbv97}).
To apply this to nearest-centroid, we write $wMM^\dagger w^\dagger$ as an inner product of tensors $\braket{a|b}$, where
\begin{align*}
    a &:= \sum_{i=1}^d\sum_{j=1}^{n+1}\sum_{k=1}^{n+1} M_{ji}\|M_{k,*}\| \ket{i}\ket{j}\ket{k} = M \otimes \tilde{M}; \\
    b &:= \sum_{i=1}^d\sum_{j=1}^{n+1}\sum_{k=1}^{n+1} \frac{w_j^\dagger w_kM_{ki}}{\|M_{k,*}\|} \ket{i}\ket{j}\ket{k}.
\end{align*}
Then, we show we have $\SQ$ access to one of the tensors ($a$).
With this, we see that the quadratic speedup from amplitude amplification is the only speedup that quantum nearest-centroid achieves:

\begin{theorem}[Classical Nearest-Centroid]
    Suppose we are given $O(T)$-time $\SQ(V) \in \BB{C}^{n\times d}$ and $\SQ(u) \in \BB{C}^d$.
    Then one can output a solution to Problem~\ref{prob:sc} in $O(T\frac{Z^2}{\eps^2}\log\frac{1}{\delta})$ time.
\end{theorem}

\section{Principal Component Analysis}
We now dequantize Lloyd, Mohseni, and Rebentrost's quantum principal component analysis (QPCA) algorithm \cite{lmr14}, an influential early example of QML \cite{bwprwl17,chdprsw18}.
While the paper describes a more general strategy for Hamiltonian simulation of density matrices, their central claim is an exponential speedup in an immediate application: producing quantum states corresponding to the top principal components of a low-rank dataset \cite{lmr14}.

The setup for the problem is as follows: suppose we are given a matrix $A \in \BB{R}^{n \times d}$ whose rows correspond to data in a dataset.
We will find the principal eigenvectors and eigenvalues of $A^\dagger A$; when $A$ is a mean zero dataset, this corresponds to the top principal components.

\begin{problem}[Principal component analysis] \label{prob:pca}
Suppose we are given access to $A \in \BB{C}^{n\times d}$ with singular values $\sigma_i$ and right singular vectors $v_i$.
Further suppose we are given $\sigma$, $k$, and $\eta$ with the guarantee that, for all $i \in [k]$, $\sigma_i \geq \sigma$ and $\sigma_i^2 - \sigma_{i+1}^2 \geq \eta\|A\|_F^2$.
With probability $\geq 1-\delta$, output estimates $\hat{\sigma}_1^2,\ldots,\hat{\sigma}_k^2$ and $\hat{v}_1,\ldots,\hat{v}_k$ satisfying $|\hat{\sigma}_i^2 - \sigma_i^2| \leq \eps_\sigma\|A\|_F^2$ and $\|\hat{v}_i - v_i\| \leq \eps_v$ for all $i \in [k]$.
\end{problem}
Denote $\|A\|_F^2/\sigma^2$ by $K$.
Lloyd et al.\ get the following:
\begin{theorem} \label{thm:qpca}
Given $\|A\|_F$ and the ability to prepare copies of $\ket{A}$ in $O(T)$ time, a quantum algorithm can output the desired estimates for Problem~\ref{prob:pca} $\hat{\sigma}_1^2,\ldots,\hat{\sigma}_k^2$ and $\ket{\hat{v}_1},\ldots, \ket{\hat{v}_k}$ in $\tilde{O}(TK\min(\eps_\sigma, \delta)^{-3})$ time.
\end{theorem}

Later results \cite[Theorems~5.2, 27]{kp16,cgj18} improve the runtime here to $\tilde{O}(TK\eps_\sigma^{-1}\polylog(nd/\delta))$ when $A$ is given in QRAM.
We will compare to the original QPCA result.

To dequantize QPCA, we use a similar high-level idea to that of the quantum-inspired recommendation system \cite{tang18a}.
We begin by using a low-rank approximation algorithm, Algorithm~\ref{alg:svd}, to output a description of approximate top singular values and vectors.

Algorithm~\ref{alg:svd} finds the large singular vectors of $A$ by reducing its dimension down to $W$, whose singular value decomposition we can compute quickly.
Then, $S, \hat{U}, \hat{\Sigma}$ define approximate large singular vectors $\hat{V} := S^\dagger \hat{U}\hat{\Sigma}^{-1}$.
The full set of guarantees on the output of Algorithm~\ref{alg:svd} are in the supplemental material~\cite{Note1}, but in brief, for the right setting of parameters, the columns of $\hat{V}$ and the diagonal entries of $\hat{\Sigma}$ satisfy the desired constraints for our $\hat{v}_i$'s and $\hat{\sigma}_i$'s in Problem~\ref{prob:pca}.
The $\hat{\sigma}_i$'s are output explicitly, but the $\hat{v}_i$'s are described implicitly: $\hat{v}_i = S^\dagger \hat{U}_{*,i} / \hat{\sigma}_{i}$.
We have $O(T)$-time $\SQ(S)$ because all rows are normalized, and rows of $S$ are simply rows of $A$.
Thus, sampling from $\tilde{S}$ is a uniform sample from $[q]$ and sampling from $S_{i,*}$ is sampling from a row of $A$.
$\hat{U}_{*,i}$ is an explicit vector, so in essence, we need $\SQ$ access to a linear combination of vectors, each of which we have $\SQ$ access to.

\begin{algorithm}[H]
\caption{Low-rank approximation~\cite{fkv98}} \label{alg:svd}
\begin{algorithmic}[1]
    \Statex {\bf Input:} $O(T)$-time $\SQ(A) \in \BB{R}^{m\times n}$, $\sigma$, $\eps$, $\delta$
    \Statex {\bf Output:} $\SQ(S) \in \BB{C}^{\ell \times n}, \Q(\hat{U}) \in \BB{C}^{q \times \ell}, \Q(\hat{\Sigma}) \in \BB{C}^{\ell \times \ell}$
    \State Set $K = \|A\|_F^2/\sigma^2$ and $q = \Theta\big(\frac{K^4}{\eps^2}\log(\frac1\delta)\big)$
    \State Sample rows $i_1,\ldots,i_q$ from $\tilde{A}$ and define $S \in \BB{R}^{q \times n}$ such that $S_{r,*} := A_{i_r,*}\frac{\|A\|_F}{\sqrt{q}\|A_{i_r,*}\|}$
    \State Sample columns $j_1,\ldots,j_q$ from $\C{F}$, where $\C{F}$ denotes the distribution given by sampling a uniform $r \sim [q]$, then sampling $c$ from $S_r$.
    \State Let $W \in \BB{C}^{q\times q}$ be the normalized submatrix $W_{*,c} := \frac{S_{*,j_c}}{q\C{F}(j_c)}$
    \State Compute the left singular vectors of $W$ $\hat{u}^{(1)},\ldots,\hat{u}^{(\ell)}$ that correspond to singular values $\hat{\sigma}^{(1)},\ldots,\hat{\sigma}^{(\ell)}$ larger than $\sigma$
    \State Output $\SQ(S)$, $\hat{U} \in \BB{R}^{q\times \ell}$ the matrix with columns $\hat{u}^{(i)}$, and $\hat{\Sigma} \in \BB{R}^{\ell\times \ell}$ the diagonal matrix with entries $\hat{\sigma}^{(i)}$.
\end{algorithmic}
\end{algorithm}

Algorithm~\ref{alg:matvec} does exactly this: it uses rejection sampling to dequantize the swap test over a subset of qubits (getting $\ket{Vw}$ via $\bra{V}(\ket{w} \otimes I)$).

\begin{algorithm}[H]
\caption{Matrix-vector SQ access} \label{alg:matvec}
\begin{algorithmic}[1]
    \Statex {\bf Input:} $O(T)$-time $\SQ(V^\dagger) \in \BB{C}^{k\times n}$, $\Q(w) \in \BB{C}^k$
    \Statex {\bf Output:} $\SQ^\nu(Vw)$
    \Function{RejectionSample}{$\SQ(V^\dagger),\Q(w)$} \label{proc:rejsamp}
    \State Sample $i \in [k]$ proportional to $|w_i|^2\|V_{*,i}\|^2$ by manually calculating all $k$ probabilities
    \State Sample $s \in [n]$ from $V_{*,i}$ using $\SQ(V^\dagger)$
    \State Compute $r_s = (Vw)_s^2/(k\sum_{j=1}^k(V_{sj}w_j)^2)$ (after querying for $w_j$ and $V_{sj}$ for all $j \in [k]$)
    \State Output $s$ with probability $r_s$ (success); otherwise, output $\varnothing$ (failure)
    \EndFunction
    \State \textsc{Query}: output $(Vw)_s$
    \State \textsc{Sample}: run \textsc{RejectionSample} until success (outputting $s$) or $kC(V,w)\log\frac1\delta$ failures (outputting $\varnothing$)
    \State \textsc{Norm}($\nu$): Let $p$ be the fraction of successes from running \textsc{RejectionSample} $\frac{k}{\nu^2}C(V,w)\log\frac1\delta$ times; output $pk\sum_{i=1}^k |w_i|^2\|V_{*,i}\|^2$
\end{algorithmic}
\end{algorithm}

\begin{restatable}[{\cite[Proposition~4.3]{tang18a}}]{proposition}{matvec}\label{prop:matvec}
For $V \in \BB{C}^{n\times k}, w \in \BB{C}^k$, given $\SQ(V^\dagger)$ and $\Q(w)$, Algorithm~\ref{alg:matvec} simulates $\SQ^\nu(Vw)$ where the time to query is $O(Tk)$, sample is $O(Tk^2C(V,w)\log\frac{1}{\delta})$, and query norm is $O(Tk^2C(V,w)\frac{1}{\nu^2}\log\frac{1}{\delta})$.
Here, $\delta$ is the desired failure probability and $C(V,w) = \sum \|w_iV_{*,i}\|^2/\|Vw\|^2$.
\end{restatable}

In general, $C(V,w)$ may be arbitrarily large, but in this application it is $O(K)$.
Quantum algorithms achieve a speedup here when $k$ is large and $C(V,w)$ is small, such as when $V$ is a high-dimensional unitary, confirming our intuition that unitary operations are hard to simulate classically.

Altogether, we get our desired result.
\begin{theorem} \label{thm:cpca}
Given $O(T)$-time $\SQ(A) \in \BB{C}^{n\times d}$, with $\eps_\sigma, \eps_v,\delta \in (0,0.01)$, there is an algorithm that output the desired estimates for Problem~\ref{prob:pca} $\hat{\sigma}_1,\ldots,\hat{\sigma}_k$ and $O(T\frac{K^9}{\eps^4}\log^3(\frac k\delta))$-time $\SQ^{0.01}(\hat{v}_1,\ldots,\hat{v}_k)$ in $O(\frac{K^{12}}{\eps^6}\log^3(\frac k\delta)+T\frac{K^{8}}{\eps^4}\log^2(\frac k\delta))$ time, where $\eps = \min(0.1\eps_\sigma K^{1.5},\eps_v^2\eta,\frac14K^{-1/2})$.
\end{theorem}
Under the non-degeneracy condition $\eta \leq \frac14K^{-1/2}$, this runtime is $\tilde{O}(T\frac{K^{12}}{\eps_\sigma^6\eps_v^{12}}\log^3(\frac{1}{\delta}))$.
While the classical runtime depends on $\eps_v$, note that a quantum algorithm must \emph{also} incur this error term to learn about $v_i$ from copies of $\ket{v_i}$.
For example, computing entries or expectations of observables of $v_i$ given copies of $\ket{v_i}$ requires $\poly(\frac{1}{\eps_v})$ or $\poly(n)$ time.

\section{Discussion}

We have introduced the SQ access assumption as a classical analogue to the QML state preparation assumption and demonstrated two examples where, in this classical model, we can dequantize QML algorithms with ease.
We now discuss the implications of this work with respect to related literature.

A natural question is of this work's relation to classical literature: does this work improve on classical algorithms for linear algebra in any regime?
The answer may be no, for a subtle but fundamental reason: recall that our main idea is to introduce an input model \emph{strong enough} to give classical versions of QML while being \emph{weak enough} to extend to settings like QRAM, where classical computers can only access the input in very limited ways.
In particular, the SQ access model that we study is \emph{weaker} than the typical input model used for classical sketching algorithms~\cite{Mahoney11,Woodruff14,kv17}.
$O(T)$-time algorithms in the quantum-inspired access model are $\tilde{O}(\textrm{nnz}+T)$-time algorithms in the usual RAM model (where $\textrm{nnz}$ is the number of nonzero entries of the input), but not vice versa: typical sketching algorithms can exploit better data structures provided they only take $O(\textrm{nnz})$ time (e.g.\ oblivious sketches), whereas the quantum-inspired model can only use the QRAM data structure.
The crucial insight of this work is that some algorithms (such as Algorithm~\ref{alg:svd} of Frieze et al.\ \cite{fkv98}) generalize to the weaker quantum-inspired model.
Our algorithms give exponential speedups in the quantum-inspired setting, but since the model is weaker, one might expect that they perform worse in typical settings for classical computation (see \cite{adbl20}).
These model considerations also explain why we use Frieze et al.~\cite{fkv98}: to our knowledge, this algorithm is the only one from the classical literature that naturally generalizes to the $\SQ$ input model.

The closest analogue to these results and techniques is a work by Van den Nest on probabilistic quantum simulation \cite{vdNest11}, which describes a notion of ``computationally tractable'' (CT) states that corresponds to our notion of $\SQ$ access for vectors.
With this notion, the author describes special types of circuits on CT states where weak simulation is possible, using variants of Propositions~\ref{prop:inner}~and~\ref{prop:matvec}.
However, Van den Nest's work does not have a version of Algorithm~\ref{alg:svd}, since this technique only runs quickly on low-rank matrices, making it ineffective on generic quantum circuits.
We exploit this low-rank structure for efficient quantum simulation of a small-but-practically-relevant class of circuits: quantum linear algebra on data with low-rank structure.
So, our techniques used for supervised clustering are within the scope of Van den Nest's work, whereas our techniques for PCA are new to this line of work.

These techniques are not new to quantum simulation in general.
Others have considered applying randomized numerical linear algebra to quantum simulation~\cite{rwcrps18}, but have not made the connection towards dequantizing quantum algorithms, especially in large generality.
Low-rank approximation is crucial for tensor network simulations of quantum systems~\cite{Schollwock11,Orus14}, where simulation can be done efficiently provided the input is, say, a matrix product state with low tensor rank.
In this context, low-rank approximation is often performed exactly and only on a subset of the space, instead of approximately done on the full state, as is done here.
This reflects the fact that tensor network algorithms assume that the system is reasonably approximated by a tensor network and aims to work well in practice, whereas our ``dequantized'' algorithms must work on a broader class of input and prioritizes provable guarantees in an abstract computational model over real-world performance.
Nevertheless, some of these dequantized algorithms might be able to be matched by tensor network contraction techniques, when the input has low \emph{tensor} rank.
See the supplemental material for further discussion of this comparison~\cite{Note1}.

Since this work, numerous follow-ups have cemented the significance of the $\SQ$ access model introduced here \cite{cglltw20,cllw20,cglltw19,jgs20}.
In particular, a recent work \cite{cglltw19} essentially dequantizes the singular value transformation framework of Gilyen et al.\ \cite{gslw18} when input is given in QRAM.
These works use fundamentally the same techniques to dequantize a wide swathe of low-rank quantum machine learning---an exciting step forward in understanding QML.


\appendix

\begin{acknowledgments}
Thanks to Ronald de Wolf for giving the initial idea to look at QPCA.
Thanks to Nathan Wiebe for helpful comments on this document.
Thanks to Daniel Liang and Patrick Rall for their help fleshing out these ideas and reviewing a draft of this document.
Thanks to Scott Aaronson for helpful discussions.
This material is based upon work supported by the National Science Foundation Graduate Research Fellowship Program under Grant No. DGE-1762114.
\end{acknowledgments}

%
\end{document}


\title{Supplemental Material for ``Quantum principal component analysis only achieves an exponential speedup because of its state preparation assumptions''}


\author{Ewin Tang}
\email[]{ewint@cs.washington.edu}
\homepage[]{ewintang.com}
\affiliation{University of Washington}

\date{\today}

\maketitle

\section{Motivating sample and query access}

In this section, we demonstrate how typical input models where efficient state preparation oracles are possible also admit efficient sample and query access to input.
Throughout, we will try to be precise about time complexities, but note that for classical algorithms we use the word RAM model of classical computation, which suppresses some $\log(n)$ factors (e.g.\ for reading indices) that are not conventionally suppressed for quantum circuit complexity.

For all these settings, the input vector $x \in \BB{C}^n$ is given \emph{classically}, meaning that there is some way to efficiently compute $x_i$ for all $i \in [n]$.
This is typical for the QML literature, and makes sense practically in the context of including QML into a classical ML pipeline.
However, we note that this assumption is crucial for our techniques, and when it is no longer satisfied (say, the input comes from a quantum system), we cannot get $\SQ$ access to input.

\paragraph{QRAM.}
Quantum random access memory is a proposal to implement state preparation of an $n$-dimensional quantum state in time polynomial in $\log n$ through the use of a clever data structure and parallelization \cite{glm08,Pre18,Pra14,Aar15}.
If the same input is reused or dynamically updated \cite{kp16}, using QRAM for QML can amortize the cost of data loading, removing the input bottleneck of these algorithms in certain cost models.
In our previous work \cite[3.1, 3.2]{tang18a}, we noted that this gives sample and query access to input as well.

\paragraph{Grover-Rudolph state preparation.}
Suppose we have an efficient way to apply an integration function on the state $x \in \BB{C}^n$ we wish to prepare: $\C{I}(s,t) = \sum_{i=s}^t |x_i|^2$.
Then, Grover-Rudolph state preparation \cite{grover02} can perform the rotations done in a QRAM protocol without issues with parallelization, only assuming the ability to query entries of $x$ in superposition.
So, one can prepare $\ket{x}$ with $O(\log(n))$ applications of the integration oracle.

Given an efficient integration oracle, we can also gain sample and query access to $x$: because all of the nodes of a QRAM can be expressed as an integral of $x$ over some interval, we can simply perform that sample and query access protocol, replacing calls to the data structure with calls to the integration oracle.
If integrating takes $O(T)$ time, this gives $O(T\log n)$-time $\SQ(x)$, the equivalent result to the quantum setting.


\paragraph{Sparsity.}
When $x \in \BB{C}^n$ has only $s$ nonzero entries, and we know their locations, we can prepare $\ket{x}$ in $O(s\log n)$ gates.
We can also get $\SQ(x)$ by querying all of the nonzero entries to compute the corresponding distribution for sampling and norm in $O(s)$ time.

\paragraph{Uniformity.}
When entries of $x$ are close to uniform, the corresponding state can be prepared given the ability to query $x$ in superposition.
Classically, samples and norm estimates can be found quickly using rejection sampling: sample $i \in [n]$ uniformly at random, and choose it with probability $\frac{n|x_i|^2}{C\|x\|^2}$, where $C$ is chosen so that $C\geq \max_{i \in n} \frac{n|x_i|^2}{\|x\|^2}$; otherwise, restart.

\paragraph{Generic state preparation.}
Finally, we can generically convert a state preparation procedure to sample and query access: given input vectors $x \in \BB{C}^n$ classically, norms $\|x\|$, and a quantum procedure to create copies of $\ket{x}$, we can get $\SQ(x)$ by measuring $\ket{x}$ in the computational basis.
This could apply to cases where preparing $\ket{x}$ is feasible, but a full QRAM algorithm is impractical due to the space and time overhead inherent to current QRAM proposals (e.g.\ from storing magic states or error-correcting).

\section{Relation to Tensor Networks}

Since this work can be seen as a form of quantum simulation, it makes sense to compare and contrast the algorithms presented here to tensor network algorithms.
Although the most direct analogue of our work, van den Nest's paper~\cite{vdNest11}, does not use low-rank approximation, it is a typical subroutine in the tensor networks context.

As a reminder, tensor network algorithms work by using tensor network representations of quantum states, such as matrix product states (MPS) or projected entangled pair states (PEPS).
For a small-but-practically-relevant set of quantum states, such as the ground states of one-dimensional quantum spin systems, this tensor network representation is highly space-efficient~\cite{vc06,Orus14}.
Since this representation is also fairly \emph{time}-efficient, particularly excelling for 1D tensor networks like MPS, finding and analyzing these ground states through techniques like the density-matrix renormalization group method (DMRG)~\cite{White92,Schollwock11} and time-evolving block decimation (TEBD)~\cite{Vidal03,Vidal04} is surprisingly tractable.
Specifically, in this 1D case, many tensor network manipulations can be done in time polynomial in the \emph{bond dimension} of the tensor network (sometimes called the \emph{tensor rank})~\cite{clopzm16}.
These tensor network algorithms still work even when the states being studied cannot be exactly represented with a small tensor network, provided they can still be closely \emph{approximated} by a tensor network of the desired shape.
This approximation is typically done through iterative low-rank approximation on local pieces of the state, as is done in the aforementioned DMRG and TEBD algorithms~\cite{Schollwock11,bc17}.
This technique is heuristic, though: since NP-hard problems can be encoded into the ground states of local Hamiltonians, algorithms that find MPS approximations to ground states must perform poorly on some inputs~\cite{Eisert06}.

We can observe that tensor network algorithms, at least naively, do not suffice to recover the results described here.
First, our algorithms require far weaker assumptions: instead of assuming the input state can be represented as an MPS with low \emph{tensor rank}, we only assume that the input mixed state has a density matrix with low \emph{linear algebraic rank}.
For example, the results presented here support simulating QPCA applied to arbitrary pure states (though this is true for trivial reasons, since the only principal component of a pure state is the pure state itself) and probability distributions over a constant number of pure states.
A tensor network algorithm might run into issues representing this pure state, since without any other restrictions, one needs an exponential number of parameters to specify it.
Of course, one could perform repeated low-rank approximation to an input state to force it into, say, an MPS, but because this is done without considering the structure of the input state, this would likely incur large error.
Our model bypasses these issues of representation by assuming an efficient representation of the input already exists in the form of $\SQ$ access: our algorithm works by observing that a representation for the output of QPCA can be found efficiently classically if it's allowed to depend on the input's representation.
With our model abstracting away representation details, these classical analogues can run in time independent of any notion of bond dimension.

In summary: our algorithms can be run efficiently on the set of low-rank states, and this set is much larger than the set of states that are efficiently represented by tensor networks.
This is not because the classical algorithms are particularly good compared to state-of-the-art tensor network methods, but because the $\SQ$ model abstracts away the difficulties that arise from representing states.
One might be able to use the tensor network literature to prove that, say, QPCA on matrix product states can be simulated classically, but these results work in a far more general regime.
This is the main advantage of the $\SQ$ model: the wide generality in which it applies.

Though tensor network algorithms and the dequantized algorithms presented here have drastically different goals, we can still compare how the differences in their desiderata affect their application and analysis of low-rank approximation.
Because tensor network algorithms perform repeated low-rank approximation on different parts of the space, the iterative algorithm that results generally has no provable guarantees, especially on general inputs, but performs exceptionally well in practice for the settings where this low-rank approximation makes sense physically~\cite{lvv15}.
The algorithms presented here are not intended to be used in practice, since (as discussed in the main text) we developed them with the intention of being barriers for asymptotic speedup for quantum machine learning algorithms.
Instead, we want provable guarantees, which our algorithm for QPCA achieves by performing low-rank approximation only once, which produces low error for precisely the class of input density matrices under consideration.

\section{Details for Nearest-Centroid Classification}

Now, we give proofs to propositions and theorems stated in the main text, beginning with the dequantization of nearest-centroid classification.
Several of these results are versions of protocols from prior work~\cite{tang18a}; we provide proofs for completeness, but encourage readers to look there for more details.

\subsection{Proof of Proposition~\ref{prop:inner}}

\begin{algorithm}[H]
\caption{Inner product estimation} \label{alg:inner}
\begin{algorithmic}[1]
    \Statex {\bf Input:} $O(T)$-time $\SQ^\nu(x) \in \BB{C}^{n}$, $\Q(y) \in \BB{C}^n$
    \Statex {\bf Output:} an estimate of $\braket{x|y}$
    \State Let $s = 54\frac{1}{\eps^2}\log\frac{2}{\delta}$
    \State Collect measurements $i_1,\ldots,i_s$ from $\ket{x}$ \label{line:meas}
    \State Let $z_j = x_{i_j}^\dagger y_{i_j}\frac{\|x\|^2}{|x_{i_j}|^2}$ for all $j \in [s]$ \Comment{$\BB{E}[z_j] = \braket{x|y}$}
    \State Separate the $z_j$'s into $6\log\frac{2}{\delta}$ buckets of size $\frac{9}{\eps^2}$, and take the mean of each bucket
    \State Output the (component-wise) median of the means
\end{algorithmic}
\end{algorithm}

\setcounter{theorem}{2}
\begin{restatable}[{\cite[Proposition~4.2]{tang18a}}]{proposition}{inner}\label{prop:inner}
For $x,y \in \BB{C}^n$, given $\SQ^\nu(x)$ and $\Q(y)$, Algorithm~\ref{alg:inner} outputs an estimate of $\braket{x|y}$ to $(\eps+\nu +\eps\nu)\|x\|\|y\|$ error with probability $\geq 1-\delta$ in time $O(\frac{T}{\eps^2}\log\frac{1}{\delta})$.
\end{restatable}

\begin{proof}
    We analyze Algorithm~\ref{alg:inner}.
    Computing the $z_i$'s, means, and medians all take linear time, so the runtime is dominated by the collection of samples, Line~\ref{line:meas}, as desired.

    As for correctness, first consider the case that $\nu = 0$.
    Then the $z_j$'s satisfy that $\BB{E}[z_j] = \braket{x|y}$ and $\Var[z_j] \leq \|x\|^2\|y\|^2$.
    Taking the mean of $\frac{9}{\eps^2}$ copies of $z_j$ reduces the variance to $\eps^2\|x\|^2\|y\|^2/9$, so by Chebyshev's inequality, the probability that the mean is $\frac{\eps}{\sqrt{2}}\|x\|\|y\|$-far from $\braket{x|y}$ is $\leq \frac{2}{9}$.

    Next, we treat the real and imaginary components of the means separately, and take the median of the means component-wise.
    To show that this median of means satisfies the desired error bound, notice that if half of the $6\log\frac2\delta$ means are within $\frac{\eps}{\sqrt{2}}\|x\|\|y\|$ of $\braket{x|y}$ in the real axis, then the median of these $6\log\frac2\delta$ means is also within $\frac{\eps}{\sqrt{2}}\|x\|\|y\|$ of $\braket{x|y}$ in the real axis.
    Let $\mathcal{E}_i$, for $i \in [6\log\frac{2}{\delta}]$, be the random variable that is 0 if the $i$th mean is within $\frac{\eps}{\sqrt{2}}\|x\|\|y\|$ of $\braket{x|y}$ in the real axis, and 1 otherwise.
    Further, let $Z_i$, for $i \in [6\log\frac{2}{\delta}]$, be independent random Bernoulli variables that are 1 with probability $\frac29$ and 0 otherwise.
    We previously established that $\Pr[\mathcal{E}_i = 1] \leq \frac29$, so by a Hoeffding-Chernoff bound~\cite{Hoeffding94}, 
    \begin{align*}
        &\Pr[\textstyle\sum_{i=1}^{6\log\frac{2}{\delta}} \mathcal{E}_i \geq 3\log\frac{2}{\delta}] \\
        &\leq \Pr[\textstyle\sum_{i=1}^{6\log\frac{2}{\delta}} Z_i \geq 3\log\frac2\delta] \\
        &\leq \Big(\Big(\frac{2/9}{1/2}\Big)^{1/2}\Big(\frac{7/9}{1/2}\Big)^{1/2}\Big)^{6\log(2/\delta)}
        \leq \frac{\delta}{2}.
    \end{align*}
    We get the same bound for the imaginary component, and combine the two to get the desired correctness property.

    When $\nu > 0$, we cannot query for $\|x\|^2$ exactly, only for an estimate $\bar{x} \in (1\pm \nu)\|x\|^2$.
    Consider the output of Algorithm~\ref{alg:inner}, where we replace the $\|x\|^2$ with $\bar{x}$, and call this output $\alpha$.
    The output satisfies $|\frac{\|x\|^2}{\bar{x}}\alpha - \braket{x|y}| \leq \eps\|x\|\|y\|$ with probability $\geq 1-\delta$.
    Re-arranging, we can achieve the bound
    \begin{align*}
        |\alpha - \braket{x|y}|
        &\leq |\alpha - \tfrac{\bar{x}}{\|x\|^2}\braket{x|y}| + |(\tfrac{\bar{x}}{\|x\|^2} - 1)\braket{x|y}| \\
        &\leq (1+\nu)\eps\|x\|\|y\| + \nu|\braket{x|y}| \\
        &\leq ((1+\nu)\eps + \nu)\|x\|\|y\|.
    \end{align*}
    Note that having $\SQ(y)$ instead of $\Q(y)$ improves the scaling by no more than constant factors.
\end{proof}

We asserted in the main text that, for this protocol, quantum algorithms can achieve a quadratic speedup via amplitude estimation but no more.
This holds because a quantum algorithm for inner product estimation scaling better than $\frac{1}{\eps}$ would imply the ability to detect the existence of a marked item with faster than $\sqrt{N}$ oracle queries, contradicting unstructured search lower bounds~\cite{bbbv97}.

\subsection{Proof of Theorem~\ref{thm:csc}}

We first recall the problem we wish to solve and present the algorithm that solves it.

\setcounter{theorem}{0}
\begin{problem}[Centroid distance] \label{prob:sc}
Suppose we are given access to $V \in \BB{C}^{n\times d}$ and $u \in \BB{C}^d$.
Estimate $\|u - \frac{1}{n}\vec{1} V\|^2$ to $\eps$ additive error with probability $\geq 1-\delta$.
\end{problem}

\setcounter{algorithm}{3}
\begin{algorithm}[H]
    \caption{Classical supervised clustering} \label{alg:csc}
\begin{algorithmic}[1]
    \Statex {\bf Input:} $O(T)$-time $\SQ(V_1,\ldots,V_n, u)$
    \Statex {\bf Output:} An estimate of $\lambda = \|u - \frac{1}{n}V\vec{1}\|^2$
    \State Achieve $O(T)$-time $\SQ(a)$, $\Q(b)$ for $a$, $b$ as in (\ref{eqn:a}), (\ref{eqn:b})
    \State Run Algorithm~\ref{alg:inner} to estimate $\braket{a|b}$ to $\eps$ error and output the result

\end{algorithmic}
\end{algorithm}

\setcounter{theorem}{3}
\begin{theorem}[Classical Nearest-Centroid] \label{thm:csc}
    Suppose we are given $O(T)$-time $\SQ(V) \in \BB{C}^{n\times d}$ and $\SQ(u) \in \BB{C}^d$.
    Then Algorithm~\ref{alg:csc} outputs a solution to Problem~\ref{prob:sc} in $O(T\frac{Z^2}{\eps^2}\log\frac{1}{\delta})$ time.
\end{theorem}

\begin{proof}
Recall from the main text that the quantity we wish to estimate can be written as $wMM^\dagger w^\dagger$, where
$$
    M := \begin{bmatrix}\bar{u} \\ \frac{1}{\sqrt{n}}\bar{V}\end{bmatrix} \text{ and }
    w := \begin{bmatrix}\|u\| & -\frac{1}{\sqrt{n}}\tilde{V}\end{bmatrix}.
$$
First, notice that we have $O(T)$-time $\SQ(M)$ and $\Q(w)$: $M$'s rows are rescaled vectors that we have SQ access to; and $\tilde{M} := [1\;\frac{1}{\sqrt{n}}\;\cdots\;\frac{1}{\sqrt{n}}]$, so we have $\SQ(\tilde{M})$ trivially; and $w$'s entries can be queried using our given access.
Next, use that $wMM^\dagger w^\dagger = \braket{a|b}$ for $a,b$ the flattened tensors
\begin{align}
    a &:= \sum_{i=1}^d\sum_{j=1}^{n+1}\sum_{k=1}^{n+1} M_{ji}\|M_{k,*}\| \ket{i}\ket{j}\ket{k} = M \otimes \tilde{M}; \label{eqn:a} \\
    b &:= \sum_{i=1}^d\sum_{j=1}^{n+1}\sum_{k=1}^{n+1} \frac{w_j^\dagger w_kM_{ki}}{\|M_{k,*}\|} \ket{i}\ket{j}\ket{k}. \label{eqn:b}
\end{align}
Given $O(T)$-time $\SQ(M)$ and $\Q(w)$, we have $O(T)$-time $\SQ(a)$ and $\Q(b)$: namely, we can sample from $a$ by sampling $j$ and $k$ from $\tilde{M}$, and then sampling $i$ from $M_{j,*}$.
Thus, we can apply Proposition~\ref{prop:inner} to estimate $w^TM^TMw$ to $\eps(4Z)$ error with probability $1-\delta$ in $O(T\frac{1}{\eps^2}\log\frac{1}{\delta})$ time, using that $\|a\| = \|M\|_F^2 = 4$ and $\|b\| = \|w\|^2 = Z$.
Rescaling $\eps$ by $4Z$ gives the desired result.
\end{proof}

\section{Details for Principal Component Analysis}

\subsection{Proof of Theorem~\ref{thm:qpca}}

We recall the problem that QPCA solves.

\setcounter{theorem}{4}
\begin{problem}[Principal component analysis] \label{prob:pca}
Suppose we are given access to $A \in \BB{C}^{n\times d}$ with singular values $\sigma_i$ and right singular vectors $v_i$.
Further suppose we are given $\sigma$, $k$, and $\eta$ with the guarantee that, for all $i \in [k]$, $\sigma_i \geq \sigma$ and $\sigma_i^2 - \sigma_{i+1}^2 \geq \eta\|A\|_F^2$.
With probability $\geq 1-\delta$, output estimates $\hat{\sigma}_1^2,\ldots,\hat{\sigma}_k^2$ and $\hat{v}_1,\ldots,\hat{v}_k$ satisfying $|\hat{\sigma}_i^2 - \sigma_i^2| \leq \eps_\sigma\|A\|_F^2$ and $\|\hat{v}_i - v_i\| \leq \eps_v$ for all $i \in [k]$.
\end{problem}

\setcounter{theorem}{5}
\begin{theorem} \label{thm:qpca}
Given $\|A\|_F$ and the ability to prepare copies of $\ket{A}$ in $O(T)$ time, a quantum algorithm can output the desired estimates for Problem~\ref{prob:pca} $\hat{\sigma}_1^2,\ldots,\hat{\sigma}_k^2$ and $\ket{\hat{v}_1},\ldots, \ket{\hat{v}_k}$ in $\tilde{O}(TK\min(\eps_\sigma, \delta)^{-3})$ time.
\end{theorem}

The quantum algorithm we use is QPCA~\cite{lmr14}, and in particular, we use Prakash's analysis of Lloyd et al's QPCA:

\begin{proposition*}[3.2.1~\cite{Pra14}] \label{prop:lmr}
    Let $\rho$ be a density matrix with eigenvalues $\lambda_i$ and eigenvectors $v_i$.
    Given the ability to prepare states with density matrix $\rho$ in $T$ time, Lloyd et al's PCA algorithm~\cite{lmr14} runs in time $\tilde{O}(T\min(\eps_\sigma, \delta)^{-3})$ and returns a sample $(\ket{v},\bar{\lambda})$ where $\Pr[\ket{v}=\ket{v_j}] = \lambda_j$ and $\bar{\lambda} \in \lambda_j \pm \eps_\sigma/4$ with probability at least $1-\delta$.
\end{proposition*}

In other words, QPCA run on states with density matrix $\frac{1}{\|A\|_F^2}A^\dagger A$ outputs a random singular vector $\ket{v_i}$ with probability proportional to its singular value $\sigma_i^2$.
As we assume our eigenvalues have an $\eta\|A\|_F^2$ gap, the precise eigenvector $\ket{v_j}$ sampled can be identified by the eigenvalue estimate.
Then, by computing enough samples, we can learn all of the eigenvalues of at least $\sigma^2$ and get the corresponding states with only $\tilde{O}(\|A\|_F^2/\sigma^2)$ overhead, giving Theorem~\ref{thm:qpca}.

\begin{proof}[Proof of Theorem~\ref{thm:qpca}]
Because we can prepare $\ket{A}$ in $O(T)$ time, we can also prepare $\rho$ with density matrix $\frac{1}{\|A\|_F^2}A^\dagger A$ in $O(T)$ time, since this is the state of the second set of qubits of $\ket{A}$.
So, we can get a sample $(\ket{v}, \bar{\lambda})$ in $O(T\min(\eps_\sigma,\delta)^{-3})$ time.
The eigenvectors and eigenvalues of $\rho$ are $v_i$ and $\sigma_i^2$, respectively, from the SVD of $A$.

By assumption, we know that the first $k$ eigenvectors have size $\geq \sigma^2$, and because we estimate our $\lambda_i$'s to $\eps_\sigma/4$ error, we can identify an eigenvector by its eigenvalue estimate.
For each $i \in [k]$, $\ket{v_i}$ appears with probability $\geq \sigma^2/\|A\|_F^2$, so we see it in $\tfrac{\|A\|_F^2}{\sigma^2} \log\frac{1}{\delta}$ samples with probability $\geq 1-\delta$.
By union bound, we can find $\ket{v_i}$ and their corresponding eigenvalue estimates for all $k$ with $O(\tfrac{\|A\|_F^2}{\sigma^2}\log\frac{k}{\delta})$ samples.
The eigenvectors are exactly correct and the eigenvalues have $\eps_\sigma\|A\|_F^2$ error, as desired: so, this is the desired output.
The runtime is the time to output $O(\tfrac{\|A\|_F^2}{\sigma^2}\log\frac{k}{\delta})$ samples.
\end{proof}

Note that the quantum algorithm crucially uses the gap assumption in Problem~\ref{prob:pca}.
Without a gap assumption, it's still possible to efficiently learn the subspace spanned by the top $k$ singular vectors. In this work, we restrict our scope to the task of finding individual singular vectors as stated \cite{lmr14}.

Also note that the QPCA result isn't particularly useful if $A$ is not low-rank: if $\rho$ is high-dimensional, we would expect that $\lambda_j$'s are around $O(1/n)$, and if we wanted to distinguish individual eigenvectors, we would also have that $\eps_\sigma = O(1/n)$, ruining any exponential speedup.

\subsection{Analysis of Low-Rank Approximation Algorithm}

\setcounter{algorithm}{1}
\begin{algorithm}[H]
\caption{Low-rank approximation~\cite{fkv98}} \label{alg:svd}
\begin{algorithmic}[1]
    \Statex {\bf Input:} $O(T)$-time $\SQ(A) \in \BB{R}^{m\times n}$, $\sigma$, $\eps$, $\delta$
    \Statex {\bf Output:} $\SQ(S) \in \BB{C}^{\ell \times n}, \Q(\hat{U}) \in \BB{C}^{q \times \ell}, \Q(\hat{\Sigma}) \in \BB{C}^{\ell \times \ell}$
    \State Set $K = \|A\|_F^2/\sigma^2$ and $q = \Theta\big(\frac{K^4}{\eps^2}\log(\frac1\delta)\big)$
    \State Sample rows $i_1,\ldots,i_q$ from $\tilde{A}$ and define $S \in \BB{R}^{q \times n}$ such that $S_{r,*} := A_{i_r,*}\frac{\|A\|_F}{\sqrt{q}\|A_{i_r,*}\|}$
    \State Sample columns $j_1,\ldots,j_q$ from $\C{F}$, where $\C{F}$ denotes the distribution given by sampling a uniform $r \sim [q]$, then sampling $c$ from $S_r$.
    \State Let $W \in \BB{C}^{q\times q}$ be the normalized submatrix $W_{*,c} := \frac{S_{*,j_c}}{q\C{F}(j_c)}$
    \State Compute the left singular vectors of $W$ $\hat{u}^{(1)},\ldots,\hat{u}^{(\ell)}$ that correspond to singular values $\hat{\sigma}^{(1)},\ldots,\hat{\sigma}^{(\ell)}$ larger than $\sigma$
    \State Output $\SQ(S)$, $\hat{U} \in \BB{R}^{q\times \ell}$ the matrix with columns $\hat{u}^{(i)}$, and $\hat{\Sigma} \in \BB{R}^{\ell\times \ell}$ the diagonal matrix with entries $\hat{\sigma}^{(i)}$.
\end{algorithmic}
\end{algorithm}

\setcounter{theorem}{8}
\begin{restatable}[{\cite[Theorem~4.4]{tang18a}}, adapted from Frieze et al.\ \cite{fkv98}]{proposition}{svd} \label{prop:svd}
Suppose we are given $O(T)$-time $\SQ(A) \in \BB{C}^{n\times d}$, a singular value threshold $\sigma$, and an error parameter $\eps \in (0,\sqrt{\sigma/\|A\|_F}/4]$.
Denote $K := \|A\|_F^2/\sigma^2$.
Then in 
$$
    O\Big(\frac{K^{12}}{\eps^6}\log^3\frac1\delta + T\frac{K^8}{\eps^4}\log^2\frac{1}{\delta}\Big)
$$
time, Algorithm~\ref{alg:svd} outputs $O(T)$-time $\SQ(S) \in \BB{C}^{q \times n}$ and $\hat{U} \in \BB{C}^{q\times \ell}, \hat{\Sigma} \in \BB{R}^{\ell \times \ell}$ (for $\ell = \Theta(\frac{K^4}{\eps^2}\log\frac1\delta)$) implicitly describing a low-rank approximation to $A$, $D := A\hat{V}\hat{V}^\dagger$ with $\hat{V} := S^\dagger\hat{U}\hat{\Sigma}^{-1}$ (notice $\operatorname{rank}D \leq \ell$).
This description satisfies the following with probability $\geq 1-\delta$:
\begin{enumerate}[label=(\alph*),noitemsep]
    \item $\sigma_\ell \geq \sigma - \eps\|A\|_F$ and $\sigma_{\ell+1} < \sigma + \eps\|A\|_F$;
    \item $\|A - D\|_F^2 \leq \|A - A_\ell\|_F^2 + \eps\|A\|_F^2$;
    \item $\sum_{i=1}^\ell |\hat{\sigma}_i^2 - \sigma_i^2| \leq \eps\|A\|_F^2/K^{1.5} = \eps\sigma^3/\|A\|_F$;
    \item $\|\hat{V} - \Lambda\|_F \leq \eps$ for some $\Lambda$ with orthonormal columns and the same image as $\hat{V}$.
\end{enumerate}
\end{restatable}
Intuitively, (a) shows that $\ell$ is the right place to truncate up to $\eps$ error, (b) is a low-rank approximation guarantee, (c) shows that our singular values are approximately correct, and (d) shows that our singular vectors are approximately orthonormal.
These additive-error bounds are weaker than the relative-error bounds common in the classical ML literature.
This seems inherent, since our $\SQ$ access assumes less about input compared to the classical literature (see discussion).
Quantum algorithms also only achieve additive-error bounds.

\begin{proof}
Because the only difference between Algorithm~\ref{alg:svd} and \textsc{ModFKV} from previous work \cite{tang18a} is the different choice of $\eps$, this result mostly follows directly from previous work.
The algorithm and runtime are given in Section~4; (a) follows from Lemma~4.5($\heartsuit$); (b) follows from Lemma~4.5($\diamondsuit$); (d) follows from Proposition~4.6.
For (c), we know that
    \[ \|A^\dagger A - S^\dagger S\|_F, \|SS^\dagger - CC^\dagger\|_F \leq \|A\|_F^2/\sqrt{q} = \eps\sigma^2/K.\]
This is stated in the proof of {\cite[Proposition~4.6]{tang18a}}; no factor of $\log(1/\delta)$ appears because it is used to amplify the probability of this equation being true.
Let $\sigma_{M,i}$ denote the singular values of $M$, respectively.
By applying the Hoffman-Wielandt inequality to both, we get that
    \[ \sqrt{\sum |\sigma_{A,i}^2 - \sigma_{S,i}^2|^2}, \sqrt{\sum |\sigma_{S,i}^2 - \sigma_{C,i}^2|^2} \leq \eps\|A\|_F^2/K^2. \]
so
\begin{align*}
    \sum_{i=1}^k |\sigma_{A,i}^2 - \sigma_{C,i}^2| &\leq \sqrt{k}\Big(\sum_{i=1}^k |\sigma_{A,i}^2 - \sigma_{C,i}^2|^2\Big)^{1/2}\\ 
    &\leq \sqrt{k}\Big(\sum |\sigma_{A,i}^2 - \sigma_{C,i}^2|^2\Big)^{1/2} \\
    &\leq \sqrt{k}\eps\|A\|_F^2/K^2 \leq \eps\|A\|_F^2/K^{1.5}.
\end{align*}
\end{proof}

\subsection{Proof of Proposition~\ref{prop:matvec}}

\setcounter{algorithm}{2}
\begin{algorithm}[H]
\caption{Matrix-vector SQ access} \label{alg:matvec}
\begin{algorithmic}[1]
    \Statex {\bf Input:} $O(T)$-time $\SQ(V^\dagger) \in \BB{C}^{k\times n}$, $\Q(w) \in \BB{C}^k$
    \Statex {\bf Output:} $\SQ^\nu(Vw)$
    \Function{RejectionSample}{$\SQ(V^\dagger),\Q(w)$} \label{proc:rejsamp}
    \State Sample $i \in [k]$ proportional to $|w_i|^2\|V_{*,i}\|^2$ by manually calculating all $k$ probabilities
    \State Sample $s \in [n]$ from $V_{*,i}$ using $\SQ(V^\dagger)$ \label{line:matvecsample}
    \State Compute $r_s = (Vw)_s^2/(k\sum_{j=1}^k(V_{sj}w_j)^2)$ (after querying for $w_j$ and $V_{sj}$ for all $j \in [k]$)
    \State Output $s$ with probability $r_s$ (success); otherwise, output $\varnothing$ (failure)
    \EndFunction
    \State \textsc{Query}: output $(Vw)_s$
    \State \textsc{Sample}: run \textsc{RejectionSample} until success (outputting $s$) or $kC(V,w)\log\frac1\delta$ failures (outputting $\varnothing$)
    \State \textsc{Norm}($\nu$): Let $p$ be the fraction of successes from running \textsc{RejectionSample} $\frac{k}{\nu^2}C(V,w)\log\frac1\delta$ times; output $pk\sum_{i=1}^k |w_i|^2\|V_{*,i}\|^2$
\end{algorithmic}
\end{algorithm}

\setcounter{theorem}{6}
\begin{restatable}[{\cite[Proposition~4.3]{tang18a}}]{proposition}{matvec}\label{prop:matvec}
For $V \in \BB{C}^{n\times k}, w \in \BB{C}^k$, given $\SQ(V^\dagger)$ and $\Q(w)$, Algorithm~\ref{alg:matvec} simulates $\SQ^\nu(Vw)$ where the time to query is $O(Tk)$, sample is $O(Tk^2C(V,w)\log\frac{1}{\delta})$, and query norm is $O(Tk^2C(V,w)\frac{1}{\nu^2}\log\frac{1}{\delta})$.
Here, $\delta$ is the desired failure probability and $C(V,w) = \sum \|w_iV_{*,i}\|^2/\|Vw\|^2$.
\end{restatable}

\begin{proof}
    We analyze Algorithm~\ref{alg:matvec}.
    We use the following easy-to-verify facts about the function \textsc{RejectionSample}: (1) $r_s \leq 1$ by Cauchy-Schwarz, so the protocol is well-defined; (2) conditioned on success, $s$ is a measurement from $\ket{Vw}$, because the denominator of $r_s$ is the probability $s$ is selected in Line~\ref{line:matvecsample}, up to normalization; (3) the function succeeds with probability $(kC(V,w))^{-1}$; (4) running it takes $O(Tk)$ time.

    \textsc{Query} clearly has the stated properties.
    (2) implies that \textsc{Sample} is correct if it succeeds, (3) implies that it has the stated failure probability, and (4) implies the stated runtime.
    (3) implies that \textsc{Norm} is correct (by a Chernoff bound), and (4) implies the stated runtime.
\end{proof}

\subsection{Proof of Theorem~\ref{thm:cpca}}

\setcounter{algorithm}{4}
\begin{algorithm}[H]
  \caption{Principal component analysis} \label{alg:cpca}
\begin{algorithmic}[1]
  \Statex {\bf Input:} $\SQ(A) \in \BB{C}^{n \times d}$, $\sigma, k, \eta$ as in Problem~\ref{prob:pca}
  \Statex {\bf Output:} $\hat{\sigma}_i$ and $\hat{v}_i$ for $i \in [k]$ as in Problem~\ref{prob:pca}
  \State Set $\eps \gets \min(\eps_\sigma K^{1.5},\eps_v^2\eta,\frac14K^{-1/2})$ and $\sigma' \gets \sigma - \eps\|A\|_F$
  \State Run Algorithm~\ref{alg:svd} with parameters $\sigma'$, $\eps$, and $\delta/k$ to find a low-rank approximation described by $\SQ(S) \in \BB{C}^{q \times n}$, $\hat{U} \in \BB{C}^{q\times \ell}$, $\hat{\Sigma} \in \BB{C}^{\ell \times \ell}$
  \State For $i \in [k]$, let $\hat{\sigma}_i := \hat{\Sigma}_{ii}$, $\hat{v}_i := S^\dagger \hat{U}_{*,i}/\hat{\Sigma}_{ii}$ \label{line:pcaout}
  \State Output $\hat{\sigma}_i^2$ and $\SQ^\nu(\hat{v}_i)$; SQ access follows from Algorithm~\ref{alg:matvec} with matrix $S^\dagger$ and vector $\hat{U}^{(i)}/\hat{\Sigma}_{ii}$.
\end{algorithmic}
\end{algorithm}

\setcounter{theorem}{7}
\begin{theorem} \label{thm:cpca}
Given $O(T)$-time $\SQ(A) \in \BB{C}^{n\times d}$, with $\eps_\sigma, \eps_v,\delta \in (0,0.01)$, Algorithm~\ref{alg:cpca} outputs the desired estimates for Problem~\ref{prob:pca} $\hat{\sigma}_1,\ldots,\hat{\sigma}_k$ and $O(T\frac{K^9}{\eps^4}\log^3(\frac k\delta))$-time $\SQ^{0.01}(\hat{v}_1,\ldots,\hat{v}_k)$ in $O(\frac{K^{12}}{\eps^6}\log^3(\frac k\delta)+T\frac{K^{8}}{\eps^4}\log^2(\frac k\delta))$ time, where $\eps = \min(0.1\eps_\sigma K^{1.5},\eps_v^2\eta,\frac14K^{-1/2})$.
\end{theorem}

Most of this theorem simply follows by combining the guarantees from Proposition~\ref{prop:svd} and Proposition~\ref{prop:matvec} to get the PCA guarantees.
The only nontrivial piece of the proof is showing that $\|v_i - \hat{v}_i\| \leq \eps_\sigma$, since the low-rank approximation algorithm gives no guarantee about controlling individual $\hat{v}_i$'s.
The \ref{prop:svd}(b) bound, even with $\eps = 0$, only implies that $\{\hat{v}_i\}$ span the same subspace as $\{v_i\}$.
So, for $r \in [k]$, we consider the rank-$r$ approximation formed by using only $\hat{v}_1,\ldots,\hat{v}_r$.
The \ref{prop:svd}(b) bound for \emph{all $r$} of these low-rank approximations give the desired bounds: using the gap assumption, the first $\hat{v}_r$ to deviate from $v_r$ will violate the bound for that rank approximation.
So, all $\hat{v}_r$ must be close to their corresponding $v_r$.

\begin{proof}
First, consider runtime.
From Proposition~\ref{prop:svd}, using that $\sigma' \geq 3\sigma/4$, we have that the runtime of Algorithm~\ref{alg:cpca} is correct.
Further, for our $\hat{v}_i$'s, applying Algorithm~\ref{alg:matvec} gives the desired runtimes, since
\begin{multline*}
    C(S^\dagger, \hat{U}_{*,i}/\hat{\Sigma}_{ii}) = \frac{\sum_j \|S_{j,*}\|^2(\hat{U}_{ji}/\hat{\Sigma}_{ii})^2}{\|\hat{v}_i\|^2} \\
    \leq \frac{\|S\|_F^2\|\hat{U}_{*,i}\|^2}{\hat{\Sigma}_{ii}^2\|\hat{v}_i\|^2}
    = \frac{\|A\|_F^2}{\hat{\Sigma}_{ii}^2\|\hat{v}_i\|^2}
    \leq \frac{\|A\|_F^2}{\sigma^2(1-\eps)}
    = O(K).
\end{multline*}
where we used Cauchy-Schwarz, that $\|S\|_F^2 = \|A\|_F^2$, $\|\hat{U}_{*,i}\| = 1$, $\hat{\Sigma}_{ii} \geq \sigma$, and Proposition~\ref{prop:svd}(d).

For correctness: by \ref{prop:svd}(a), $\sigma_{\ell+1} < \sigma$, so $\ell \geq k$, and Line~\ref{line:pcaout} of Algorithm~\ref{alg:cpca} is well-defined.
Further, by \ref{prop:svd}(c), we can bound $\hat{\sigma}_i$ error $|\hat{\sigma}_i^2 - \sigma_i^2| \leq \eps\frac{\sigma'^3}{\|A\|_F} \leq \eps_\sigma\|A\|_F^2$.

Now, we show that $\|v_i - \hat{v}_i\| \leq \eps_\sigma$.
Let $\hat{V}_r$ denote $\hat{V}$ truncated to its first $r$ columns.
Observe that for all $r \leq \ell$, $\|A - A\hat{V}_r\hat{V}_r^\dagger\|_F^2 \leq \|A - A_r\|_F^2 + \eps\|A\|_F^2$ holds with probability $\geq 1-\delta$.
This holds for $r = \ell$ by Proposition~3(b), so suppose $r < \ell$.
We can imagine that our run of Algorithm~3 is equivalent to simultaneous runs with different singular value thresholds $\gamma_r := \frac12(\sigma_r + \sigma_{r+1})$, the only differences being that we truncate to threshold $\sigma$ instead of $\gamma_r$, and $q$ is larger.
Truncating to $\gamma_r$ is equivalent to truncating at the $r$th singular value, since $\hat{\sigma}_r > \gamma_r > \hat{\sigma}_{r+1}$:
\begin{align*}
  |\hat{\sigma}_i - \sigma_i| &\leq \frac{1}{\sigma}|\hat{\sigma}_i^2 - \sigma_i^2| \leq \frac1\sigma\sum_{i=1}^\ell |\hat{\sigma}_i^2 - \sigma_i^2| \leq \frac{\eta}{10}\|A\|_F; \\
  |\sigma_i - \gamma_i| &= \frac{1}{2}|\sigma_i - \sigma_{i+1}| \geq \frac{1}{4\|A\|_F}|\sigma_i^2-\sigma_{i+1}^2| \geq \frac{\eta}{4}\|A\|_F.
\end{align*}
So, $\hat{\sigma}_r \geq \sigma_r - \frac{\eta}{10}\|A\|_F > \sigma_r - \frac{\eta}{4}\|A\|_F \geq \gamma_r$, and similarly for the other side.
Increasing $q$ is equivalent to decreasing $\eps$, so for our truncated $\hat{V}_{[r]}$'s, the guarantees from Proposition~3 hold.
Namely, 3(b) holds as desired.
By choosing failure probability $\delta/k$, we can guarantee, with probability $\geq 1-\delta$, $\|A - A\hat{V}_r\hat{V}_r^\dagger\|_F^2 \leq \|A - A_r\|_F^2 + \eps\|A\|_F^2$ holds for all $r \in [k]$.

Fix $\Lambda$ as the isometry from Proposition~3(d) satisfying $\|\hat{V} - \Lambda\|_F \leq \eps$.
Now, consider a truncation $\hat{V}_r$ for any $r \in [k]$, and notice that $\|\hat{V}_r - \Lambda_r\|_F \leq \eps$.
\begin{align*}
    \|A - A\Lambda_r\Lambda_r^\dagger\|_F^2
    &\leq (\|A - A\hat{V}_r\hat{V}_r^\dagger\|_F + O(\eps)\|A\|_F)^2 \\
    &\leq \|A - A\hat{V}_r\hat{V}_r^\dagger\|_F^2 + O(\eps)\|A\|_F^2 \\
    &\leq \|A - A_r\|_F^2 + O(\eps)\|A\|_F^2
\end{align*}
Now, we rewrite the expression, heavily using that $\Lambda_r$ and $V_r$ (the matrix with columns $v_1,\ldots,v_r$ from $A$'s SVD) are isometries.
\begin{align*}
    &O(\eps)\|A\|_F^2 \\
    & \geq \|A - A\Lambda_r\Lambda_r^\dagger\|_F^2 - \|A - A_r\|_F^2 \\
    &= \|A\|_F^2 - \|A\Lambda_r\Lambda_r^\dagger\|_F^2 - (\|A\|_F^2 - \|A_r\|_F^2) \\
    &= \|A_r\|_F^2 - \|A\Lambda_r\Lambda_r^\dagger\|_F^2 \\
    &= \sum_{i=1}^r \sigma_i^2 - \sum_{i=1}^{\min m,n}\sigma_i^2\|\Lambda_r\Lambda_r^\dagger v_i\|^2 \\
    &= \sum_{i=1}^r \sigma_i^2(1-\|\Lambda_r\Lambda_r^\dagger v_i\|^2 ) - \sum_{i=r+1}^{\min m,n}\sigma_i^2\|\Lambda_r\Lambda_r^\dagger v_i\|^2 \\
    &\geq \sigma_r^2\sum_{i=1}^r (1-\|\Lambda_r\Lambda_r^\dagger v_i\|^2 ) - \sigma_{r+1}^2\sum_{i=r+1}^{\min m,n}\|\Lambda_r\Lambda_r^\dagger v_i\|^2 \\
    &= \sigma_r^2\|(I-\Lambda_r\Lambda_r^\dagger)V_rV_r^\dagger\|_F^2 - \sigma_{r+1}^2\|\Lambda_r\Lambda_r^\dagger(I-V_rV_r^\dagger)\|_F^2 \\
    &= \sigma_r^2(r-\|\Lambda_r\Lambda_r^\dagger V_rV_r^\dagger\|_F^2) - \sigma_{r+1}^2(r-\|\Lambda_r\Lambda_r^\dagger V_rV_r^\dagger\|_F^2) \\
    &\geq \eta\|A\|_F^2(r-\|\Lambda_r\Lambda_r^\dagger V_rV_r^\dagger\|_F^2)
\end{align*}
Here, we used the Pythagorean theorem ($\|A\Pi\|_F^2 + \|A(I-\Pi)\|_F^2 = \|A\|_F^2$ for $\Pi$ an orthogonal projector), that $\|\Lambda\Lambda^\dagger v_i\|^2 \in [0,1]$, that $\|VM\|_F = \|M\|_F$ if $V$ is an isometry, that Frobenius norm decomposes $\|M\|_F^2 = \sum \|M_{*,i}\|^2$, and the gap assumption from the PCA problem statement.

So, $\|\Lambda_r\Lambda_r^\dagger V_rV_r^\dagger\|_F^2 = \|\Lambda_r^\dagger V_r\|_F^2 \geq r - O(\eps/\eta)$, and equivalently, $\|(V_rV_r^\dagger - I)\Lambda_r\Lambda_r^\dagger\|_F^2 = O(\eps/\eta)$.
\begin{align*}
    &\|v_r - \hat{v}_r\| \\
    &\leq \|v_r - \Lambda_{*,r}\| + \|\Lambda_{*,r} - \hat{v}_r\|\\
    &\leq \|v_r - \Lambda_{*,r}\| + \eps \\
    &\leq \|v_r - V_kV_k^\dagger\Lambda_{*,r}\| + \|V_kV_k^\dagger \Lambda_{*,r} - \Lambda_{*,r}\| + \eps \\
    &\leq \|v_r - V_kV_k^\dagger\Lambda_{*,r}\| + \|(V_kV_k^\dagger - I)\Lambda_k\Lambda_k^\dagger\|_F + \eps\\
    &\leq \|v_r - V_kV_k^\dagger\Lambda_{*,r}\| + O(\sqrt{\eps/\eta})\\
    &= \sqrt{(v_r - V_kV_k^\dagger\Lambda_{*,r})^\dagger(v_r - V_kV_k^\dagger\Lambda_{*,r})} + O(\sqrt{\eps/\eta})\\
    &\leq \sqrt{2 - 2\langle v_r, \Lambda_{*,r}\rangle} + O(\sqrt{\eps/\eta})\\
    &= O(\sqrt{\eps/\eta}) = O(\eps_v)
\end{align*}
The last line follows from combining previously-seen inequalities:
let $x_{i,j} = \langle v_i, \Lambda_{*,j}\rangle^2$.
Then we can rewrite
\begin{align*}
    \|\Lambda_s^\dagger V_s\|_F^2 \geq s - O(\eps/\eta) &\iff \sum_{i,j = 1}^s x_{i,j} \geq s-O(\eps/\eta); \\
    \|\Lambda^\dagger v_i\|^2 \leq 1 &\iff \sum_{j=1}^k x_{i,j} \leq 1; \\
    \|(\Lambda_{*,j})^\dagger V_k\|^2 \leq 1 &\iff \sum_{i=1}^k x_{i,j} \leq 1.
\end{align*}
Now, consider taking a linear combination of these inequalities: add together the first inequality for $s = r-1, r$ and subtract the second and third inequalities for $i,j \in [r-1]$.
\begin{align*}
    \sum_{i,j = 1}^{r-1} x_{i,j} + \sum_{i,j = 1}^r x_{i,j} &\geq 2r-1-O(\eps/\eta) \\
    -\sum_{i=1}^{r-1}\sum_{j=1}^k x_{i,j} &\geq 1-r \\
    -\sum_{j=1}^{r-1}\sum_{i=1}^k x_{i,j} &\geq 1-r \\
    \implies x_{rr} - \sum_{i=1}^{r-1}\sum_{j=r+1}^k(x_{i,j} + x_{j,i}) &\geq 1 - O(\eps/\eta)
\end{align*}
So, $x_{rr} \geq 1-O(\eps/\eta)$, implying that $\langle v_r,\Lambda_{*,r}\rangle \geq 1-O(\eps/\eta)$.
\end{proof}

%